\documentclass[a4paper, 10pt, conference]{ieeeconf}      

\IEEEoverridecommandlockouts                              
\usepackage{mathtools}

\def\includegraphics{}
\newtheorem{theorem}{Theorem}[section]
\newtheorem{lemma}[theorem]{Lemma}

\usepackage{amsmath}
\usepackage[linesnumbered,ruled]{algorithm2e}
\usepackage{comment}
\usepackage{graphicx}
\usepackage{subcaption}
\usepackage{hyperref}
\usepackage{cite}
\usepackage{enumerate}
\usepackage{graphicx,caption}
\newtheorem{definition}{Definition}[section]
\usepackage{paralist} 
\usepackage{amssymb}
\usepackage{balance}
\usepackage{multicol}
\usepackage{multirow}
\usepackage{makecell}
\usepackage{chngcntr}
\usepackage{apptools}
\AtAppendix{\counterwithin{corollary}{subsection}}

\usepackage{color}

\definecolor{darkgrn}{rgb}{0, 0.75, 0}

\newcommand{\dels}{\texttt{$\Delta$-screening}}
\newcommand{\delsab}{\texttt{$\Delta$S}}

\newcommand{\louvain}{\texttt{Louvain}}
\newcommand{\slm}{\texttt{SLM}}

\newcommand{\dlouvain}{\texttt{dLouvain-$\Delta$s}}
\newcommand{\dslm}{\texttt{dSLM-$\Delta$s}}

\newcommand{\dlouvainbs}{\texttt{dLouvain-base}}
\newcommand{\dslmbs}{\texttt{dSLM-base}}

\newcommand{\Rset}{$\mathcal{R}_t$}
\newcommand{\Ci}{$C_{t-1}(i)$}
\newcommand{\Cjstar}{$C_{t-1}(j_*)$}

\newcommand{\etal}{{\it et al.}}

\title{\LARGE \bf
A Fast and Efficient Incremental Approach toward Dynamic Community Detection
}

\author{Neda Zarayeneh$^{1}$ and Ananth Kalyanaraman$^{1}$\thanks{$^{1}$N. Zarayeneh and A. Kalyanaraman are with the 
School of Electrical Engineering and Computer Science, Washington State
        University, Pullman, WA 99163, USA.
       Email:  {\tt\small neda.zarayeneh@wsu.edu, ananth@wsu.edu}}
}

\begin{document}

\maketitle
\thispagestyle{empty}
\pagestyle{empty}


\begin{abstract}
Community detection is a discovery tool used by network
scientists to analyze the structure of real-world networks. It seeks to
identify natural divisions that may exist in the input networks that
partition the vertices into coherent modules (or communities). While this
problem space is rich with efficient algorithms and software, most of this
literature caters to the static use-case where the underlying network does
\emph{not} change. However, many emerging real-world use-cases give rise to a
need to incorporate dynamic graphs as inputs. 

In this paper, we present a fast and efficient incremental approach toward
dynamic community detection. The key contribution is a generic technique
called \dels, which examines the most recent batch of changes made to an
input graph and selects a subset of vertices to reevaluate for potential
community (re)assignment. This technique can be incorporated into any of the
community detection methods that use modularity as its objective function for
clustering. 
For demonstration purposes, we incorporated the technique into two well-known
community detection tools. 
Our experiments demonstrate that our new incremental approach is able to
generate performance speedups without compromising on the output quality
(despite its heuristic nature). For instance, on a real-world network with 63M temporal edges (over 12 time steps), our approach was able to complete in 1056 seconds, yielding a 3$\times$ speedup over a baseline implementation. In addition to demonstrating the performance benefits, we also show how to
use our approach to delineate appropriate intervals of temporal resolutions at
which to analyze an input network.
\end{abstract}

\section{Introduction}
\label{sec:Intro}

Community detection is a fundamental problem in many graph applications. The goal of community detection is to identify tightly-knit groups of vertices in an input network, such that the members of each ``community'' share a high concentration of edges among them than to the rest of the network. 
Owing to its ability to reveal natural divisions that may exist in a network (in an unsupervised manner), community detection has become one of the fundamental discovery tools in a network scientist's toolkit. 
The operation is widely used in a variety of application domains including (but not limited to) social networks, biological networks, internet and web networks, citation and collaboration networks, etc. 

Designing efficient algorithms and implementations for community detection has been an area of active research for well over a decade. 
While theoretical formulations are known to be NP-Hard \cite{brandes2008modularity}, there are a number of efficient heuristics and related software already available. 
A comprehensive review of community detection methods and related applications is available in \cite{fortunato2010community}.
However, most of the existing tools target static networks, where the input graph cannot change.
On the other hand, most real-world networks are dynamic in nature, where vertices and edges can be added and/or removed over a period of time. 

Owing to the increasing availability of dynamic networks, the problem of dynamic community detection has become an actively researched topic of late, and
multiple methods have been proposed over the last decade (e.g., \cite{aktunc2015dynamic,greene2010tracking,zakrzewska2015dynamic}). 
In Section~\ref{sec:Related} we present a brief review of such related works.
Despite these advances,
a key remaining challenge in the design of these algorithms is in quickly identifying
the parts of the graph that are likely to be impacted by a change (or
collectively by a
recent batch of changes), so that it
becomes possible to update the community information with minimal
recomputation effort.

{\bf Contributions:}
In this paper, we propose an algorithmic technique and a corresponding incremental approach that would complement the developments made in dynamic community methods, and in particular those that use the modularity function \cite{newman2004fast} as their clustering objective. 
More specifically, the main contributions are as follows:
\begin{compactenum}[i)]\itemsep=-0.05ex
\item
We visit the problem of identifying vertex subsets that are likely to be impacted by the most recent batch of changes made to the graph. 
To address this problem, we present a technique called \dels, which can be efficiently implemented and incorporated as part of existing dynamic community algorithms that use modularity. 
\item 
To demonstrate and evaluate this technique, we incorporated the technique
into two well-known classical community detection methods---namely, the
Louvain method \cite{blondel2008} and the SLM method
\cite{waltman2013smart}---thereby generating two incremental clustering
implementations.
\item 
Using these two implementations, we present a thorough experimental
evaluation on both synthetic and real-world inputs. Our results show that the
\dels{} technique is effective in pruning work (to reduce recomputation
effort) without compromising on output quality.
\item 
In addition to demonstrating its performance benefits, we also show how to
use our approach to delineate appropriate intervals of temporal resolutions
at which to analyze an input network.
\end{compactenum}


\section{Related Work}
\label{sec:Related}

Algorithms to compute dynamic communities over time-evolving graphs can be broadly classified into two types.

One class of methods follows a \emph{two-step} strategy of first
identifying the best set of communities for the current time step and
then subsequently mapping them onto the communities from previous
generations to track evolution. 
Hopcroft \etal{} \cite{hopcroft2004tracking} present a method in which a
static community detection tool is individually applied to the graphs at all time
steps and the results are later combined by computing community
similarities between successive steps. 
Greene \etal{} \cite{greene2010tracking} propose a variation of this
approach where they use a matching-based formulation and a related heuristic to
map the communities of the latest time step to the communities of
previous generations (as identified by a ``front'').

In general, the two-step strategy
is better suited if the magnitude and/or complexity of changes to the
input graph is more drastic or random.  However, the strategy suffers
drawbacks in two ways: It can make tracking of communities difficult as
an independent application of static community detection at each time
step may \emph{not} necessarily preserve previous generations
communities, and therefore the outputs can become non-deterministic.
Secondly, these approaches could become expensive to run on large inputs
as the approach entails significant recomputation for each time step, and
the community tracking (via mapping) process can also be expensive.  

The other class of methods for detecting dynamic communities
follows a more \emph{incremental} strategy where communities from the
previous generation(s) are propagated and updated using changes reflected
in the current time step. 
Maillard \etal{} \cite{maillard2009modularity} propose a modularity-based
incremental approach extending upon the classical Clauset-Newman-Moore
static method \cite{clauset2004finding}. 
Aktunc \etal{} \cite{aktunc2015dynamic} propose a method DSLM as an
extension of its static predecessor \cite{waltman2013smart}. 
Xie \etal{} \cite{xie2013labelrankt} present an incremental method based
on label propagation which is a fast heuristic. 
Saifi and Guillaume \cite{seifi2012community} provide a way to track and update community ``cores'' across time steps.
Zakrzewska and Bader \cite{zakrzewska2015dynamic} present another variant 
that tracks the communities of a selected
set of seed  vertices in the graph. 
The FacetNet approach introduced by Lin \etal{} \cite{lin2008facetnet} is
a hybrid approach that operates with a dual objective of maximizing
modularity for the current time step while trying to also preserve as
much of the previous generation communities. 

In general, the incremental strategy has the advantage in runtime
(because of the reuse of community information from previous steps), and it also
has the advantage of outputting a relatively stable set of communities
across time steps. 
The technique of \dels{} proposed in this paper is aimed at helping these
incremental methods to be able to quickly identify the relevant parts of
the graph that are potentially impacted by a recent batch of changes, so
that the computation effort in the incremental step can be
reduced without compromising on clustering quality.

We note here that most of the existing methods use modularity or one of
its variants as the objective function for optimizing community
structure.
Bassett \etal{} \cite{bassett2013robust} propose and evaluate the choice
of alternative null hypothesis models in modularity-based dynamic
community detection methods.

\section{Method}
\label{sec:Method}

\subsection{Basic Notation and Terminology}
\label{sec:Notation}

A \emph{dynamic graph} $G(V,E)$ can be represented as a sequence of graphs $G_{1}(V_{1},E_{1}),\ldots, G_{T}(V_{T},E_{T})$, where $G_t(V_t,E_t)$ denotes the graph at time step $t$; 
we use $n_t=|V_t|$ and $M_t=|E_t|$.
In this paper, we consider only undirected graphs. 
The graphs may be weighted---i.e., each edge $(i,j)\in E_t$ is associated with a numerical positive weight
$\omega_{ij}\geq 0$;
if the graphs are unweighted, then the edges are assumed to be associated with unit weight,
without loss of generality.
We denote the neighbors of a vertex $i$ as
$\varGamma(i)=\{j\;|\;(i,j)\in E_t\}$. 
We use $m_t$ to denote the sum of the weights of all edges in $G_t$---i.e., $m_t=\sum_{(i,j)\in
E_t}\omega_{ij}$.  
We denote the degree of a vertex $i$ by $d(i)$. 
The \emph{weighted degree} of a vertex $i$, denoted by $d_\omega(i)$, 
is the sum of weights of all edges incident on $i$.

In this paper, we consider incrementally growing dynamic graphs, where edges and vertices can be
added (but not deleted) from one time step to another. This implies that $V_t\supseteq
V_{t-1}$ and $E_t\supseteq E_{t-1}$, for all $1<t\leq T$. 
We denote the newly added edges at any time step $t$ as $\Delta_t=E_t\setminus E_{t-1}$.

We denote the set of communities detected at time step $t$ as $\mathcal{C}_t$.
Note that, by definition, $\mathcal{C}_t$ represents a partitioning of the vertices in $V_t$---i.e., 
each community $C\in\mathcal{C}_t$ is a subset of $V_t$; 
all communities in $\mathcal{C}_t$ are pairwise disjoint; and 
$\bigcup_{C\in\mathcal{C}_t} C=V_t$.

For any vertex $i\in V_t$, we denote the community containing $i$, at any point in the
algorithm's execution, as $C(i)$,
following the convention used in \cite{lu2015parallel}.
Also, let $e_{i\rightarrow C}$ denote the sum of the weights for the edges linking vertex $i$ to
vertices in community $C$---i.e., $e_{i\rightarrow C}=\sum_{j\in
C\cap\varGamma(i)}\omega_{ij}$. 
Furthermore, let $a_C$ denote the sum of the weighted degrees of all vertices in $C$---i.e., 
$a_C=\sum_{i\in C}d_\omega(i)$.


Given the above, the \emph{modularity}, $Q_t$, as imposed by a community-wise partitioning
$\mathcal{C}_t$ over $G_t$, is given by \cite{newman2004fast}:
\begin{equation}\label{eqn:Q}
	Q_t =\frac{1}{2m_t} (\sum_{i\in V_t} e_{i\rightarrow C(i)} - \frac{1}{2m_t} \sum_{C \in \mathcal{C}_{t}} a_{C}^2)  
\end{equation}



Given a community-wise partitioning on an input graph, the \emph{modularity gain} that can be achieved by moving a particular vertex $i$ from its current community to another target community (say $C(j)$) can be calculated in constant time \cite{blondel2008}.
We denote this modularity gain by $\Delta Q_{i\rightarrow C(j)}$.

\subsection{Problem Statement}
\label{sec:problemstatement}

\begin{definition}{\bf  Dynamic Community Detection:}
{\it 
	Given a dynamic graph $G(V,E)$ with $T$ time steps, the goal of
	dynamic community detection is to detect an output set of communities
	$\mathcal{C}_t$ at each time step $t$, that maximizes the modularity
	$Q_t$ for the graph $G_t(V_t,E_t)$.
}
\end{definition}

Since the static version of the modularity optimization problem is NP-Hard \cite{brandes2008modularity}, it immediately follows that the dynamic version is also intractable.

For the static version, a number of efficient heuristics have been developed (as surveyed in \cite{fortunato2010community}). 
These approaches can be broadly classified into three categories:
divisive approaches \cite{girvan2002community,gregory2008fast}, agglomerative approaches \cite{newman2004fast,clauset2004finding}, and multi-level approaches
\cite{de2011generalized,waltman2013smart,karypis1998fast}. 
Of these, the multi-level approaches have demonstrated to be fast and effective at producing high-quality partitioning in practice. 
In Algorithm ~\ref{alg:multi-level} we show a generic algorithmic pseudocode for this class of approaches. 
While they vary in the specific details of how each step is implemented, they share several
common traits (note that this description is for the static use-case):
\begin{itemize}\itemsep=-0.05ex
\item 
At the start of each level, all vertices are assigned to a distinct community id. 
\item
An iterative process is initiated, in which 
all vertices are visited (in some arbitrary order) within each iteration, and a decision is
		made on whether to keep the vertex in its current community, or to migrate it to
		one of its neighboring communities. 
		This decision is typically made in a local-greedy fashion. For instance, 
		in the \louvain{} algorithm \cite{blondel2008}, a vertex
		migrates to a neighboring community that maximizes the modularity
		gain of that vertex---i.e., let $j\in \varGamma(i)\cup \{i\}$. Then, 
		\[C(i) \gets \arg\max_{C(j)} \Delta Q_{i\rightarrow C(j)}\]
\item
	When the net modularity gain resulting from an iteration drops below a certain
		threshold $\tau$, the current level is terminated (i.e., intra-level convergence),
		and the algorithm compacts the graph into a smaller graph by using the
		information from the communities. This procedure represents a graph coarsening
		step, and the coarsened graph is subsequently processed using the same
		iterative strategy until there is no longer an appreciable modularity gain
		between successive levels.
\end{itemize}

Algorithm~\ref{alg:multi-level} succinctly captures the main steps of the multi-level approaches.

\begin{algorithm}
	\SetAlgoNoLine
	\KwIn{$G(V,E)$}
	\KwOut{An assignment $\Pi: V \rightarrow \mathbb{Z}$}
	Initialize $\Pi$ by setting $\mathcal{C}(v) \gets \{v\},\; \forall v\in V$  \\
	\SetEndCharOfAlgoLine{}
	\Repeat{Convergence based on $Q$}{\Repeat{ Convergence based on $Q$}{\For{each $v \in V$}{
        	Compute a local (greedy) function $g(v,\mathcal{C}(v))$ \\
        	$C(v) \gets$ Update community assignment for $v$ using the results from $g(v,\mathcal{C}(v))$\\
        	}
        	Compute a global quality function $Q$ for $\Pi$ }
        Review communities of $\Pi$ (optional step)\\
        $G(V,E) \leftarrow $ Perform graph compaction for next level}
	\Return $\Pi$
	\caption{Abstraction for Multi-level Approaches}
	\label{alg:multi-level}
\end{algorithm}

\subsection{A Naive Algorithm for Dynamic Community Detection}
\label{sec:naive}
A simple approach to address the dynamic community detection problem is to directly apply the
static version of the algorithm (Algorithm~\ref{alg:multi-level}) on the graph at every time
step. However, such an approach suffers from multiple limitations. 
First, it completely ignores the communities identified at the previous time steps. This can
cause outputs to become non-deterministic, as static approaches typically prefer a random
ordering of vertices.
Furthermore, by ignoring the previous community information, the algorithm is essentially
forced to recompute from scratch, and as a result evaluate the community affiliation for \emph{all} vertices at each time step. This can be wasteful in computation.
Consider an edge $(i,j)\in\Delta_t$ that has been newly introduced at time step $t$;
it is reasonable to expect that only those vertices in the ``vicinity'' of
this newly added edge to be impacted by this addition. However, the naive strategy is not suited to exploit such proximity information, thereby negatively impacting performance particularly for large real-world networks where event-triggered changes tend to happen in a more localized manner at different time steps.

\subsection{An Incremental Approach via \dels}
\label{sec:incremental}
Here, we present an alternative approach in which we first identify a subset of
vertices to evaluate at the start of every time step, using the changes $\Delta_t$.
The idea is to identify all (or most) those vertices whose community
affiliation could potentially change
due to $\Delta_t$; the remaining vertices will simply retain their previous community assignments.
This new filtering technique, which we call \dels{} (and abbreviated as \delsab), is generic enough to be applied to any incremental clustering approach that uses modularity.
For the purpose of this paper, we demonstrate it on multi-level approaches. 

More specifically,
let $Static(G)$ denote any static community detection algorithm of choice, that takes in an
input (static) graph $G$ and outputs a set of communities $\mathcal{C}$.
Then, our incremental approach 
is as follows.
\begin{compactenum}
\item At $t=1$, the algorithm simply calls $Static(G_1)$ to output $\mathcal{C}_1$. 
\item For each subsequent time step $t>1$: 
	\begin{compactenum}[a)]
	\item The algorithm initially assigns each pre-existing vertex $i\in V_t\cap V_{t-1}$ to
		the same community label as $\mathcal{C}_{t-1}(i)$.  Each of the remaining
		vertices (i.e., those that were newly added at $t$) is assigned a distinct
		(new) community label.
	\item Next, the algorithm calls a function \dels$(G_t,\Delta_t)$ that
		returns a subset of vertices $\mathcal{R}_t\subseteq V_t$.  This subset
		corresponds to the set of vertices that have been selected for processing
		during time step $t$. 
	\item Subsequently, the algorithm calls $Static\Delta S(G_t, \mathcal{R}_t)$, which is
		a variant of $Static(G_t)$ that loads $G_t$ but visits
		only the vertex subset $\mathcal{R}_t$ for evaluation during each
		iteration---i.e., a modification to the \texttt{for} loop of line \#4 in Algorithm~\ref{alg:multi-level}.
		Note that this procedure that uses
		$\mathcal{R}_t$ is only relevant to the iterations at the first level, as in
		the subsequent levels, the algorithm uses compacted versions of the same graph. 
	\end{compactenum}

\end{compactenum}

To demonstrate the \dels{} technique, we modified two well-known community detection methods:  \louvain{} algorithm \cite{blondel2008}, and smart local moving (\slm) algorithm ~\cite{waltman2013smart}.
We call the resulting modified implementations as \dlouvain{} and \dslm{} respectively. 

Note that there is also a simpler incremental version that can be implemented
for both these methods---by following all steps outlined in our incremental
approach \emph{except} for \dels{} and instead trivially setting
$\mathcal{R}_t=V_t$.  For a comparative assessment of the \dels{} strategy,
we implemented this \emph{baseline} version as well---we refer to the
resulting two implementations as \dlouvainbs{} and \dslmbs{}\footnote{We note
that our \dslmbs{} implementation is in effect same as \cite{aktunc2015dynamic}.} respectively.

\subsection{The \dels{} Scheme}
\label{sec:VS}
In what follows, we describe our \dels{} scheme in detail. 
Given the graph ($G_t$) and changes ($\Delta_t$) at time step $t$, 
the goal of \dels{} is to identify a vertex subset $\mathcal{R}_t\subseteq V_t$ for
reevaluation at time step $t$---i.e., any vertex that is added to $\mathcal{R}_t$ will be
evaluated for potential migration by the iterative clustering algorithm
(Algorithm~\ref{alg:multi-level}); all other vertices are not evaluated (i.e., they retain
their respective community assignment from the previous time step $t-1$).
Our \dels{} scheme is also a heuristic and it does not guarantee the
reproduction of the results from the corresponding baseline version---\dlouvainbs{} or \dslmbs{}---which are also heuristics. 
The main objective here is to save runtime by reducing the number of vertices to process, 
without significantly altering the quality. Despite its heuristic nature,
however, our \dels{} scheme is designed to preserve the key behavioral traits of the baseline version (as we show in lemmas later in this section).

{\bf Algorithm:}
We assume that $\Delta_t$ is stored as a list of \emph{ordered} pairs of the form $(i,j)$. This
implies that for each newly added edge $(i,j)$, there will be two entries stored in
$\Delta_t$: $(i,j)$ and $(j,i)$, as the input graph is undirected.
We refer to the first entry ($i$) of an ordered pair ($(i,j)$) as the ``source'' vertex and
the other vertex ($j$) as the ``sink''. 
Let $S_\Delta$ denote the set of all source vertices in $\Delta_t$, and
$T_\Delta(i)$ denote the set of all sinks for a given source $i$.

Algorithm~\ref{alg:VS} shows the algorithm for \dels. 
We initialize $\mathcal{R}_t$ to $\emptyset$. 
Subsequently, we examine all edges of $\Delta_t$ in the sorted order of its source vertices. 
Sorting helps in two ways: 
It helps us to consider all the new edges incident on a given source vertex
collectively and identify the edge that (locally) maximizes the net
modularity gain (consistent with line \#4 of Algorithm~\ref{alg:VS}). This way we are able to mimic the behavior of the baseline versions which also use the same greedy scheme to migrate vertices. 
Furthermore, this sorted treatment helps reduce overhead by helping to update
$\mathcal{R}_t$ in a localized manner (relative to the source vertices) and
avoiding potential duplications in the computations associated with a vertex.

Once sorted, we read the adjacency list for each source vertex
(Algorithm~\ref{alg:VS}:line \#3), identify a neighbor ($j_*$) that maximizes the modularity gain (line \#4), and update $\mathcal{R}_t$ based on that vertex (line \#8).
However, prior to updating $\mathcal{R}_t$, we check if the selected vertex $j_*$ has a better incentive to move to $i$'s community \Ci{} (line \#7); 
if that happens, then $\mathcal{R}_t$ is not updated from source $i$ and
instead, that decision is left/deferred until $j_*$ is visited as the source. 
This way we avoid making conflicting decisions between source and sink, while decreasing the time for processing (by reducing $\mathcal{R}_t$ size). 
Note that we only use the direction of migration that results in the larger of the two gains for updating $\mathcal{R}_t$. The decision to migrate itself is deferred until the stage of execution of the iterative algorithm. In other words, the \dels{} procedure does \emph{not} modify the state of communities, but it sets the stage for which communities to be visited during the main iterative process.

\begin{algorithm}
	\SetAlgoNoLine
	\KwIn{$G_t$, $\Delta_{t}$}
	\KwOut{$\mathcal{R}_t$: Subset of vertices for reeavaluation }
	$\mathcal{R}_t \leftarrow \emptyset$\\
	Sort edges in $\Delta_t$ based on the source\\
	\For{each $i\in S_\Delta$} {
	  Let $j_*\gets {\arg\max}_{j\in T_\Delta(i)} \{\Delta Q_{i\rightarrow
	  \mathcal{C}_{t-1}(j)} \} $\\
	  Let $gain_1 \gets \Delta Q_{i\rightarrow \mathcal{C}_{t-1}(j_*)}$\\
	  Let $gain_2 \gets \Delta Q_{j_*\rightarrow \mathcal{C}_{t-1}(i)}$\\
	  \If{$gain_1\geq gain_2$ and $gain_1>0$} {
		$\mathcal{R}_t \leftarrow \mathcal{R}_t \cup \{i,j_*\} \cup \varGamma(i) \cup
		\mathcal{C}_{t-1}(j_*)$\\
	  }
	}
	\Return $\mathcal{R}_t$
\caption{\dels{} at time step $t$}
	\label{alg:VS}
\end{algorithm}

The main part of Algorithm~\ref{alg:VS} is on line \#8, where $\mathcal{R}_t$ is updated. 
Our scheme adds the following subset to $\mathcal{R}_t$:
vertices $i$ and $j_*$, all neighbors of $i$ ($\varGamma(i)$), and all vertices in the community containing $j_*$. 
In what follows, using a combination of lemmas, we show that the
$\mathcal{R}_t$ so constructed is positioned to capture all (or most of the)
``essential'' vertices that are likely to be impacted by the edge additions in $\Delta_t$.
In other words, if a vertex is not added to $\mathcal{R}_t$, it can be
concluded that it is less likely (if at all) to be impacted by the changes to
the graph, and therefore it can stay in its previous community state---thereby saving runtime.

\begin{figure}[thpb]
	\centering
	\captionsetup{width=.9\linewidth}
	\framebox{\parbox{3in}{\includegraphics[width= \linewidth]{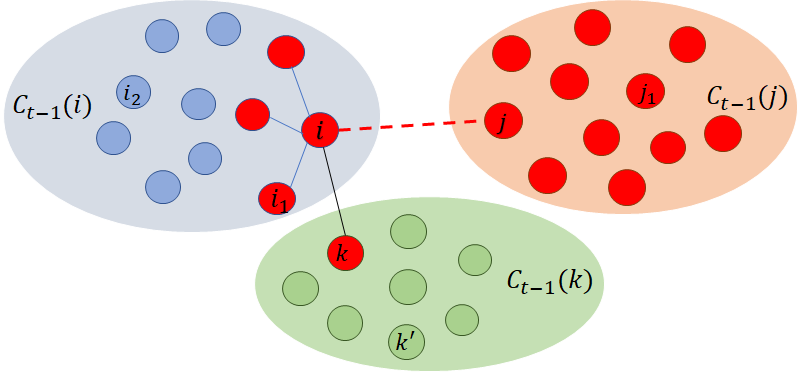}}}
	\caption{
	Figure showing the impact of a newly added edge $(i,j)$, shown in red dotted line. 
	Representative cases of candidate vertices for potential inclusion in $\mathcal{R}_t$ are shown highlighted as labelled vertices. 
	Note that we follow the naming convention of denoting the community containing a vertex $i$ by $C(i)$. Also, note that not all edges are shown. 
	}
	\label{fig:VS}
\end{figure}

In all these lemmas, for sake of convenience (and without loss of generality), we analyze the potential impact of the event represented by moving $i$ to $j_*$'s community. 
Intuitively, the key to populating $\mathcal{R}_t$ is in anticipating which
vertices are likely to alter their community status triggered by this
migration event.
Fig.~\ref{fig:VS} shows the different representative cases that originate for consideration in our lemmas.

First, we claim that any vertex that is a neighbor of $i$ can be potentially impacted. 
\begin{lemma}\label{lemmaGammai}
If $i^\prime\in\varGamma(i)$, then the community state for $i^\prime$ could potentially alter at time step $t$ if $i$ migrates to $C(j_*)$.
\end{lemma}
\begin{proof}
There are two subcases:
(A) if $i^\prime$ is also in \Ci; and (B) otherwise. 

Subcase (A) is represented by vertex label $i_1$ in Fig.~\ref{fig:VS}. 
If $i$ were to leave \Ci, the strength of the connection of $i_1$ to \Ci{} can only weaken because of a decrease in the positive term of the modularity (Eqn.~\ref{eqn:Q})). 
Even if the negative term of the same equation also decreases (due to departure of $i$ from \Ci, it may or may not be sufficient to keep $i_1$ in \Ci.
Therefore, we add $i_1$ to \Rset.

Subcase (B) is represented by vertex label $k$ in Fig.~\ref{fig:VS}. 
Here, $k$ is in a community different from \Ci.
However, the situation with $k$ is similar to that of $i_1$ in subcase (A), as $k$'s connection to its present community could potentially weaken if it discovers a stronger connection to $C(j_*)$ as a result of $i$'s move.
Therefore, we add $k$ to \Rset.
\end{proof}

Next, we analyze the potential of vertices that are in \Ci{} but \emph{not} in $\varGamma(i)$ to be impacted by the migration of $i$. In fact, we conclude that there is no need to include such vertices in \Rset.

\begin{lemma}\label{lemmaCi}
	If $i^\prime\in$ \Ci, then at time step $t$, a change to the community state of $i^\prime$ is possible, \emph{only if} $i^\prime$ is also a neighbor of $i$.
\end{lemma}
\begin{proof}
We have already considered the case where $i^\prime\in\varGamma(i)$ (as part of Lemma~\ref{lemmaGammai}). 
Therefore we only need to consider the case where $i^\prime\notin\varGamma(i)$.
This is represented by vertex label $i_2$ in Fig.~\ref{fig:VS}.
Since $i_2$ is already in \Ci{} and since $i_2$ does not share an edge with $i$, a departure of $i$ from \Ci{} can only positive reinforce $i_2$'s membership in \Ci. 
This can be shown more formally by comparing the modularity gains associated with $i_2$. 
Owing to space limitations, we show the expanded proof in Appendix: Section ~\ref{sec:proof_lemma3.2}
in \cite{1904.08553}.
In summary, vertex $i_2$ will have little incentive to change community status and therefore can be excluded from \Rset. 
\end{proof}

Next, we analyze the potential impact of $i$'s migration on members of  $j_*$'s community.  

\begin{lemma}\label{lemmaCj}
	If $j_1\in C_{t-1}(j_*)$, then at time step $t$, a change to  the community status of any such $j_1$ is possible. 
\end{lemma}
\begin{proof}
Regardless of whether $j_1$ shares a direct edge with $j_*$ or not, the migration of a new vertex ($i$) into its present community ($C_{t-1}(j_*)$) increases the negative term in Eqn.~\ref{eqn:Q}. This may or may not be accompanied with an increase in the positive term as well (depending on whether $j_1$ shares an edge with the incoming vertex $i$). In either case, however, we need to re-evaluate the community status of such vertices. 
Therefore, we add $j_1$ to \Rset.
\end{proof}

Finally, we analyze the impact of $i$'s potential migration from \Ci{} to \Cjstar, on vertices that are in neither of those two communities \emph{and} are also not in $\varGamma(i)$.

\begin{lemma}\label{lemmaCk}
If $k\in V_t\backslash \{C(i)\cup C(j)\}$, then at time step $t$, unless $k$ is also in $\varGamma(i)$, there is no need to include $k$ in \Rset.
\end{lemma}
\begin{proof}
We consider only vertices $k\notin\varGamma(i)$, as the other case was already covered in Lemma~\ref{lemmaGammai}. 
There are three subcases:
(A) $k$ shares an edge with some vertex in \Ci{} except $i$;
(B) $k$ shares an edge with some vertex in \Cjstar; and
(C) $k$ has no neighbors in \Ci{} or \Cjstar.
However, in none of these cases a migration of $i$ to \Cjstar, could create an incentive for $k$ to move to \Cjstar. 
This is shown formally in Appendix: Section ~\ref{sec:proof_lemma3.4}
in \cite{1904.08553}.
\end{proof}

\section{Experimental Evaluation}
\label{sec:Exper}

\subsection{Experimental Setup}
\label{sec:ExSetup}

{\bf Input data:}
For experimental testing, we used a combination of synthetic and real-world networks. 
Table~\ref{tab:data_stat} shows the input statistics for the inputs used.

\begin{table}[!ht]
\caption{
Input network statistics. 
See Fig. B.2 
of Appendix in \cite{1904.08553} for more details on individual time steps.
}
\label{tab:data_stat}
\centering
\resizebox{\columnwidth}{!}{%
\begin{tabular}{ |c|c|r|r|r| }
\hline
Input & Input graph & No. vertices& No. edges (cumulative) & No. timesteps\\ \hline
\multirow{4}{*}{\thead{Synthetic}} & 50k\_ll & 50,000& 2,362,448 &10 \\ 
 & 50k\_hh & 50,000 & 2,367,024 & 10\\
 & 5M\_ll & 5,000,000 & 213,656,492 &10\\
 & 5M\_hh & 5,000,000 & 213,771,700 &10\\ \hline
\multirow{2}{*}{\thead{Real-world}} & Arxiv HEP-TH
 & 27,770 & 352,807 &11 \\
 & sx-stackoverflow & 2,601,977 & 63,497,050 &2-28 \\
 \hline
\end{tabular}
}
\end{table}
As synthetic inputs, we used a collection of streaming networks 
available on the MIT Graph Challenge 2018 \cite{kao2017streaming}. 
We used two types of networks: 
i) Low block overlap, Low block size variation (abbreviated as ``\texttt{ll}''), and 
ii) High block overlap, High block size variation (abbreviated as ``\texttt{hh}''). 
These two types are in the increasing order of their community complexity
(\texttt{ll}$<$\texttt{hh}). 
However, in both cases, the number of edges grows linearly with time step (see
Appendix Fig. B.2
in \cite{1904.08553}). 
The datasets are available from sizes of 1K nodes to 20M nodes, and each of these datasets has ten time steps. 
For our testing purposes, we used the 50K and 5M datasets. 

As real-world inputs, we used two networks downloaded from SNAP database
\cite{snapnets}: 
\begin{enumerate}
    \item \textbf{Arxiv HEP-TH: }
This is a citation graph for 27,770 papers (vertices) with 352,807 edges
(cross-citations).
Even though the edges are directed, for the purpose of
analysis in this paper we treated them as undirected. 
The dataset covers papers published between 1993 and 2003. 
Consequently, we partitioned this period into 11 time steps (one for each
year). 
\item \textbf{sx-stackoverflow: } 
This is a temporal network of interactions on Stack Overflow, with $2,601,977$
vertices (users) and $63,497,050$ temporal edges (interactions). Interactions
could be one of many types---e.g.,  a user answers another user's query, a
user commented on another user's answers, etc. 
We treated all these interactions equivalently (as edges), and for the purpose
of our analysis we used only the first instance of a user-user interaction as an edge.
\end{enumerate}

{\bf Implementations tested:}
In our experiments, we tested the following implementations:
\begin{compactenum}\itemsep=-0.05ex
\item{\it Static:} 
This is a (static) community detection code run from scratch on the
graph at each time step $i$. 
\texttt{Louvain} \cite{blondel2008} and \texttt{SLM} \cite{waltman2013smart} 
are the two tools we used for this purpose.
\item{\it Baseline:}
This is a community detection code run \emph{incrementally} on the
graph at each time step $i$.
``Incremental'' here implies that at the start of every time step $i$,
we initialize the state of communities to that of the end of the previous time step $i-1$
(for $i>0$).
For this purpose, we implemented our own incremental version of the \texttt{Louvain} tool---we 
call this \texttt{dLouvain-base}); and
for SLM, we use the already available incremental version DSLM \cite{aktunc2015dynamic}---we 
call this \texttt{dSLM-base}).

\item{\it \dels:}
This is a modified baseline version incorporated with our \dels{} step
to identify the \Rset{} set for use within each time step. 
The corresponding two implementations are referred to as
\texttt{dLouvain-\delsab} and \texttt{dSLM-\delsab}.

\end{compactenum}

\subsection{Runtime and Quality Evaluation}
\label{sec:Exp-runtime}

\begin{figure}[tbh]
    \centering
    \captionsetup{width=.9\linewidth}
    	\begin{subfigure}{0.32\textwidth}
		 \includegraphics[width=\textwidth]{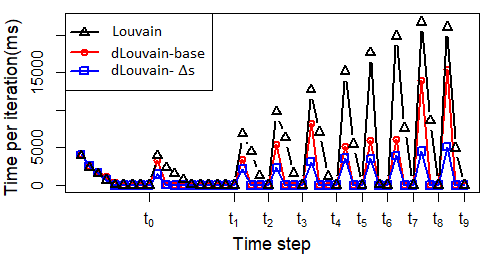}
		\caption{Average time per iteration for $5M\_{ll}$}
		\label{fig:itr_5M_ll}
	\end{subfigure}
	\begin{subfigure}{0.32\textwidth}
		 \includegraphics[width=\textwidth]{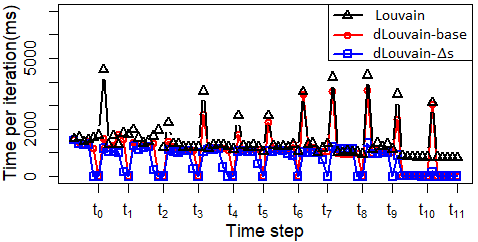}
		\caption{Average time per iteration for sx-stackoverflow}
		\label{fig:itr_sx}
	\end{subfigure}
	\caption{
The variation of average time per iteration by time step, for two
representative inputs: 5M\_ll, and sx-stackoverflow.
The average is given by the mean time to execute an iteration within
each level of a time step. 
All runs reported are from the \texttt{Louvain}-based implementations.
}
\label{fig:iter}
\end{figure}

\begin{figure}[!ht]
	\centering
	\captionsetup{width=.9\linewidth}
     \includegraphics[width= .7\linewidth]{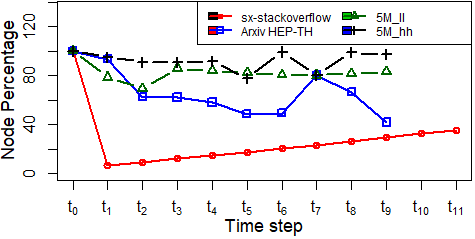}
	\caption{
	The fraction of vertices processed at every iteration, given
	by $\frac{|\mathcal{R}_t|}{|V_t|}$.
	A lower percentage corresponds to a larger savings in performance.
	}
	\label{fig:R_perc}
\end{figure}

First, we evaluate the impact of \dels{} technique on clustering
algorithm's performance. 
Since the main part of the algorithm is the iterative loop that scans
every vertex and assigns communities (the \texttt{for} loop in 
Algorithm~\ref{alg:multi-level}), we measured the average time taken
per iteration of a given level, and the results are plotted in 
Fig.~\ref{fig:iter}.
As can be observed, \dels{} achieves a significant reduction in the time
spent within each iteration (compared to both static and baseline).

\begin{figure*}[!ht]
\centering
    \captionsetup{width=.9\linewidth}
    \begin{subfigure}{0.32\textwidth}
		 \includegraphics[width=\textwidth]{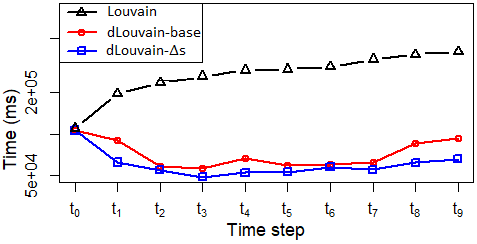}
		\caption{$5M\_{ll}$ (using Louvain) }
			\label{fig:run_5M_ll}
	\end{subfigure}
	\begin{subfigure}{0.32\textwidth}
		 \includegraphics[width=\textwidth]{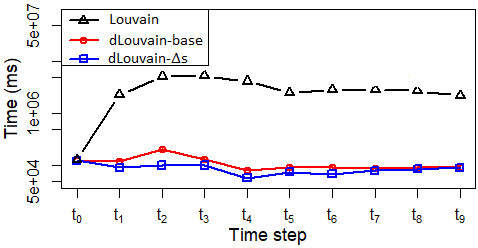}
		\caption{$5M\_{hh}$ (using Louvain)}
			\label{fig:run_5M_hh}
	\end{subfigure}
	\begin{subfigure}{0.32\textwidth}
		 \includegraphics[width=\textwidth]{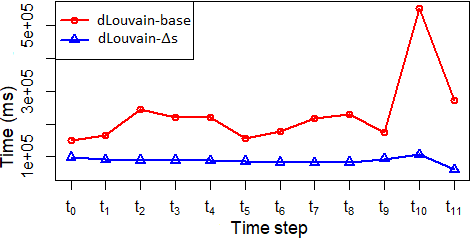}
		\caption{sx-stackoverflow (using Louvain)}
		\label{fig:run_sx_louv}
	\end{subfigure}
	\begin{subfigure}{0.32\textwidth}
	   \includegraphics[width=\textwidth]{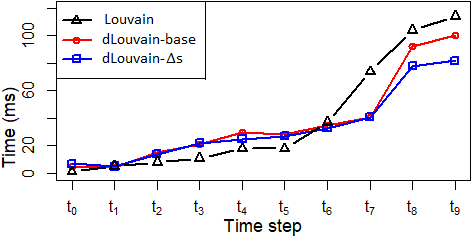}
		\caption{Arxiv HEP-TH (using Louvain)}
		\label{fig:run_cite_louv}
	\end{subfigure}
	\begin{subfigure}{0.32\textwidth}
		 \includegraphics[width=\textwidth]{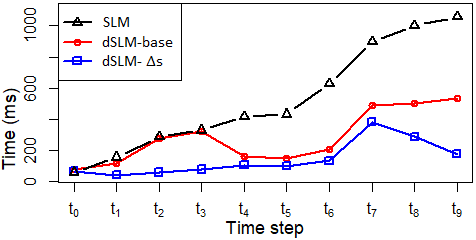}
		\caption{Arxiv HEP-TH (using SLM)}
		\label{fig:run_cite_slm}
	\end{subfigure}
	%
	\caption{
		Runtime as a function of the time steps for various inputs,
		running either Louvain- or SLM-based implementations. The
		static Louvain implementation's run and the static/baseline SLM runs 
		on sx-stackoverflow experienced long runtimes and therefore we
		omit those curves.
	}
	\label{fig:runtime}
\end{figure*}
The savings are a result of the reductions in the number of vertices processed per
iteration in the \delsab{} version (i.e., \Rset{} set size). 
We found the \Rset{} set-based savings to be more significant for the
real-world inputs compared to the synthetic. 
This is shown in Fig.~\ref{fig:R_perc}, which shows the
percentage of vertices processed per iteration---the \Rset{} set
size fractions range from as little as under 10\% in some time steps
(for real-world inputs) to as much as 100\% in some time steps (for
the synthetic inputs). 
This wide variation in efficacy can be attributed to the nature of input
changes. 
Even though both classes of input graphs (synthetic and real-world)
show linear growth rates in size, 
for the synthetic inputs it is harder to benefit from \dels{} because
edges are connected almost randomly from new to existing vertices
(as introduced by an edge sampling randomized procedure \cite{kao2017streaming}); 
whereas, in the real-world networks, changes happen more in a
localized manner, giving an opportunity to benefit from \dels.
In fact, even between the two real-world networks, we observed a significant
difference in the filtering efficacy of the \dels{} technique. 
More specifically, with the ArXiv HEP-TH input, the \Rset-size fraction
varied between $\approx$50\%-90\%; whereas the savings were much more
significant in the case of sx-stackoverflow (and also showing a linear trend).

Next, we evaluate the total runtime including the time taken to
execute all levels. Note that in multi-level codes, the number of
iterations per level may vary across the different implementations. 
Fig.~\ref{fig:runtime} shows the runtime as a function of the time
steps, for different combinations of four inputs (5M\_ll, 5M\_hh, Arxiv HEP-TH,
sx-stackoverflow) and two sets of implementations (Louvain and SLM).

We find that in all cases both baseline and \delsab{} implementations
consistently outperform the respective static implementation,
providing more than 2 orders of magnitude in some cases. 
Between the baseline and \delsab{} implementations, the difference varies
based on the input. For the synthetic inputs, both versions perform
comparably with a slight advantage to the \delsab{} implementation in some
time steps.  As discussed earlier, 
this can be attributed to the random nature of changes induced in the synthetic inputs. 
For the real-world inputs, \delsab{} significantly outperforms baseline,
yielding over 5X speedup in some time steps. For example, in Fig. \ref{fig:run_sx_louv} we have a  5X speedup in time step $t_{10}$, since in time step $t_9$ we have 55158947 edges and this number has increased to 60714297 in time step $t_{10}$ and most of the new edges are intra-community edges which result in less number of nodes for reevaluation and consequently less time.

We also evaluated the quality (measured by modularity) achieved by each
version. In nearly all cases, we observed that the \delsab{} version yielded
almost the same modularity as the baseline version despite its heuristic
nature in selecting subsets of vertices for evaluation within each iteration. 
Owing to space limitations, we show these results in Appendix 
Fig. B.3 
in \cite{1904.08553}.

\subsection{Effect of Varying the Temporal Resolution}
\label{sec:tem_bin}

\begin{figure}[!ht]
    \centering
    \captionsetup{width=.8\linewidth}
    	\begin{subfigure}{0.32\textwidth}
		 \includegraphics[width=\textwidth]{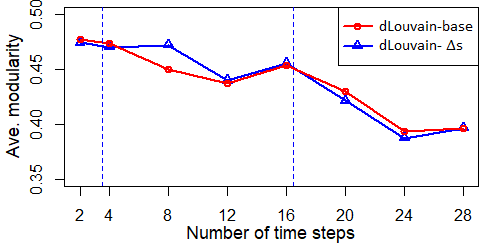}
	\caption{The change in average modularity achieved for different
	temporal resolutions (input: sx-stackoverflow)}
		\label{fig:ave_mod}
	\end{subfigure}
	\begin{subfigure}{0.32\textwidth}
		   \includegraphics[width=0.95\textwidth]{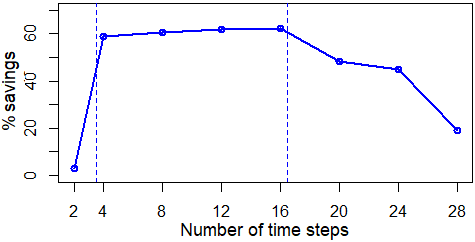}
		\caption{Percentage savings in total time for \dels{}, compared to baseline, for different temporal resolutions (input: sx-stackoverflow)}
		\label{fig:percentage_saving}
	\end{subfigure}
	\caption{
	Plots showing the effect of varying the temporal resolution---as measured
	by the number of temporal bins (i.e., time steps) used to partition
	the input graph stream. 
	The resolution of partitioning changes from coarser to finer, from
	left to right on the x-axis. 
	}
	\label{fig:sx}
\end{figure}

In many real-world use-cases, even though the input graph is available as a
temporal stream, the appropriate temporal scale to analyze them is not known
{\it a priori}. In fact, this scale is an input property that a domain
expert expects to discover through the analysis of dynamic communities. 
In order to facilitate such a study through dynamic community detection,
in this section, we study the effect of varying the \emph{temporal
resolution}, as defined by the number of time steps used to partition a graph
stream, on the output clustering. 

More specifically, using the sx-stackoverflow input, we first generated
multiple temporal datasets, each of which representing the input stream
divided into a certain number of time steps, ranging from  
\{2, 4, 8, 12, 16, 20, 24, 28\} steps. 
Note that in this scheme, there are multiple nested hierarchies---for instance, the
16-time steps partitioning can be achieved by splitting each of the 8-time
steps partitions into two. 
Subsequently, we ran our \delsab-enabled incremental clustering
method on the different temporal input datasets (we used dLouvain-\delsab{}
for this analysis).

Fig.~\ref{fig:sx} shows the results of our analysis. 
Fig.~\ref{fig:sx}a shows the change in average modularity as we increase the
temporal resolution from coarser to finer (left to right along the x-axis). 
We observe that the modularity values decline gradually until around 16 time
steps, after which the decline starts to accelerate. 
The decline in modularity suggests that the community-based structure of the
underlying network (at different scales) starts to weaken as we increase the
temporal resolution. This is to be expected as the temporally binned
graphs tend to only become sparser with increasing resolution.
Notably, the more rapid slide that starts to appear after the 16 time
steps-resolution suggests that the community structure starts to deteriorate
after that resolution for \emph{this} input.

Interestingly, this property is better captured by Fig.~\ref{fig:sx}b, which
shows the \% savings  in total runtime achieved by our \dels{} filtering scheme. Intuitively, when the \% savings
remains approximately steady (highlighted by the plateau region from
the resolution of 4 time steps to 16 time steps), it means that the nature of the
evolution of the graphs within those resolutions is also relatively consistent.
However, a steeper decline (on either end of the plateau) suggests that under
those temporal resolution scales the temporally partitioned graphs become
either too sparse (right) or too dense (left). 

Note that software speed becomes an important enabling factor
for conducting these types of experiments, where one needs to run the dynamic
community analysis repeatedly under different configurations.

\section{Conclusion}
\label{sec:Conclu}

Conducting community detection-based analysis on large dynamic networks is a
time-consuming problem and there have been many incremental strategies
proposed. In this paper, we visit a subproblem in this context---one of
identifying vertices that are likely to be impacted by a new batch of
changes. We presented a generic technique called \dels{} that examines and
selects provably ``essential'' vertices for evaluation during the $i^{th}$
time step based on the loci of the changes. Subsequently we incorporated this
technique into two widely-used community detection tools (Louvain and SLM)
and demonstrated speedups in performance without compromising on the output
quality, for a collection of synthetic and real-world inputs.  
Future research directions include: i) extension of the \dels{} technique to
edge deletions; ii) parallelization on multicore platforms; iii)
extensions to other incremental community detection tools; and iv)
application and dynamic community characterization on large-scale real-world
networks.





\section*{ACKNOWLEDGMENT}
This research was supported by U.S. National Science Foundation grant 1815467.




\newpage
\appendix
\label{sec:appendix}
  
\subsection{Expanded Versions of Proofs}
\label{sec:app_proofs}

\subsubsection{For Lemma~3.2}
\label{sec:proof_lemma3.2}

{\bf Lemma:}
        If $i^\prime\in$ \Ci, then at time step $t$, a change to the community state of $i^\prime$ is possible, \emph{only if} $i^\prime$ is also a neighbor of $i$. 

\begin{proof}
We have already considered the case where $i^\prime\in\varGamma(i)$ (as part of Lemma~\ref{lemmaGammai}). 
Therefore we only need to consider the case where $i^\prime\notin\varGamma(i)$.
This is represented by vertex label $i_2$ in Fig.~\ref{fig:VS}.
Since $i_2$ is already in \Ci{} and since $i_2$ does not share an edge with $i$, a departure of $i$ from \Ci{} can only positive reinforce $i_2$'s membership in \Ci. 
More formally, this can be shown by comparing the old vs. new modularity gains, ${\Delta Q}_{i_2\rightarrow C(j_*)}$, resulting from moving $i$ to $C_{t-1}(j_*)$:  
 \begin{equation*}
 \begin{split}
   {\Delta Q}^{\textrm{new}}_{i_2\rightarrow C_{t-1}(j_*)}  =& \frac{e_{i_2\rightarrow C_{t-1}(j_*) } - e_{i_2\rightarrow C_{t-1}(i) \backslash \{i_2\}}}{2m_t} \\
   +& \frac{d_\omega(i_2).(a_{C_{t-1}(i) \backslash \{i_2\}} - d_\omega(i)) }{(2m_t)^{2}} \\
   -&\frac{d_\omega(i_2) .(a_{C_{t-1}(j_*)} + d_\omega(i))}{(2m_t)^{2}}\\
   =& {\Delta Q}^{\textrm{old}}_{i_2\rightarrow C_{t-1}(j_*)}  - \frac{2 d_\omega(i_2). d_\omega (i))}{(2m_t)^{2}}
   \end{split}
\end{equation*}
Since the new modularity gain is less than the old, vertex $i_2$ will have no incentive to change community status, and therefore can be excluded from \Rset.
\end{proof}

Note that the above proof only analyzes the direct impact that vertex $i$'s migration from \Ci{} will have on vertices of \Ci{} but not in $\varGamma(i)$.
However, there could still be an indirect impact (cascading from Lemma~\ref{lemmaGammai})---e.g.,
the migration of a vertex $i_1\in\varGamma(i)$ may trigger the migration of a vertex $i_2\notin\varGamma(i)$ but in $\varGamma(i_1)$, and so on.
However, as this distance from the locus of change increases, the likelihood of its impact is expected
to decay rapidly in practice (as observed in our experiments).

\subsubsection{For Lemma~3.4}
\label{sec:proof_lemma3.4}

{\bf Lemma:}
If $k\in V_t\backslash \{C(i)\cup C(j)\}$, then at time step $t$, unless $k$ is also in $\varGamma(i)$, there is no need to include $k$ in \Rset.

\begin{proof}
We consider only vertices $k\notin\varGamma(i)$, as the other case was already covered in Lemma~\ref{lemmaGammai}.
There are three subcases:
(A) $k$ shares an edge with some vertex in \Ci{} except $i$;
(B) $k$ shares an edge with some vertex in \Cjstar; and
(C) $k$ has no neighbors in \Ci{} or \Cjstar.
However, in none of these cases a migration of $i$ to \Cjstar, create an incentive for $k$ to move to \Cjstar.
This can be shown more formally using modularity gains as follows.


\begin{enumerate}[(A)]
\item
                We can follow the same logic in lemma ~\ref{sec:proof_lemma3.4}. If $k$ is a neighbor of $ i $ which are located in the other communities like $k$ in Fig. ~\ref{fig:VS}, the modulaity gain ${\Delta Q}_{ k \rightarrow C_{t-1}(j_*)}$ will increase; therefore $k$ should be considered for reevaluation.
\item
                 Any node $k$ connected to \Cjstar{} which is located in other communities have the following modularity gain after moving $i$ to \Cjstar{}{}:
                 \begin{equation*}
                 \begin{split}
           {\Delta Q}^{\textrm{new}}_{k\rightarrow C_{t-1}(j_*)} \approx & \frac{e_{ k \rightarrow C_{t-1}(j_*)} - e_{k\rightarrow C_{t-1} \backslash \{k \}}}{2m_t} \\
           +& \frac{d_\omega(k).a_{C_{t-1}(k) \backslash \{k\}}}{(2m_t)^{2}}\\
           -&\frac{d_\omega(k).(a_{C_{t-1}(j_*)} + d_\omega (i))}{(2m_t)^{2}}\\
           \approx & {\Delta Q}^{\textrm{old}}_{k \rightarrow C_{t-1}(j_*)}  - \frac{d_\omega(k).d_\omega (i)}{(2m_t)^{2}}
           \end{split}
	   \end{equation*}
                 which means the modularity gain will decrease and $k$ will not change its community membership.  
\item
                If a node $k$ is not connected to either \Ci{} or \Cjstar{} like $k\prime$ in Fig. ~\ref{fig:VS}, the modularity gain of $k$ moving to any communities will not change and therefore the node will not reconsider moving to other community.
\end{enumerate}
\end{proof}

\subsection{Figures}
\label{sec:app_figures}
\counterwithin{figure}{subsection}
\setcounter{figure}{0}  
\begin{figure}[tbh]
        \centering
        \captionsetup{width=.9\linewidth}
        \framebox{\parbox{3in}{\includegraphics[width=\linewidth]{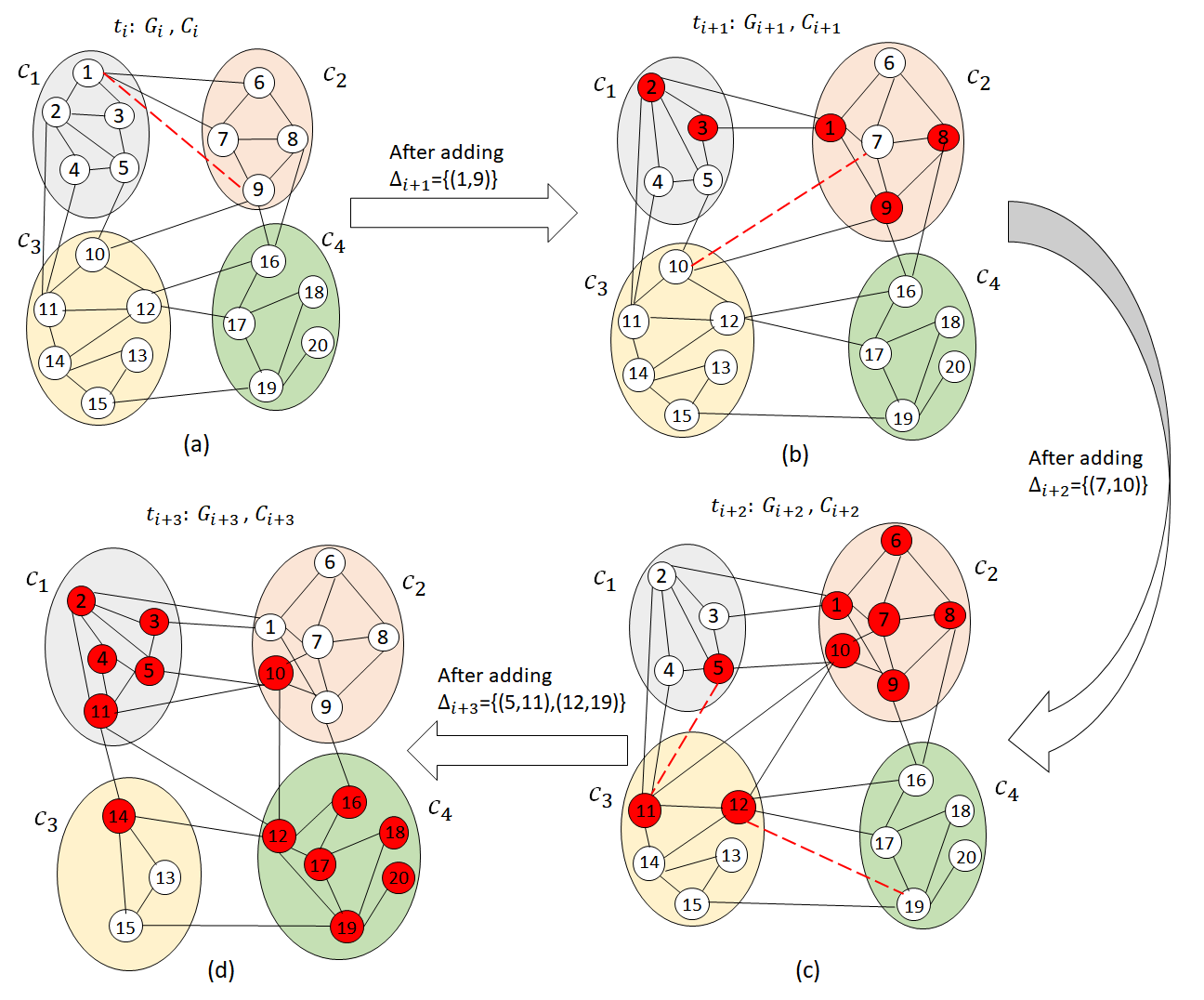}}}\qquad
        \caption{This is an example of \dels{} method for three time steps. First, by adding the edge $(1,9)$ to $G_i$ as $\Delta Q_{(1 \rightarrow c_2)} > 0$, our heuristic method will change the community membership of node $1$ to $c_2$ and the subset of nodes for reevaluation will be $R=\{1,2,3,8,9\}$ which are shown in red in part (b). Then, the edge $(7,10)$ is added with $\Delta Q_{(10 \rightarrow c_2)}>0$;therefore, the community membership of node $10$ will be changed to $c_2$ and consequently, the subset of nodes for reevaluation will be $R=\{1,5,6,7,8,9,10,11,12\}$ which is displayed in red in part (c). Finally, two edges $\{(5,11),(12,19)\}$ are added and as $\Delta Q_{(5 \rightarrow c_1)}>0$, $\Delta Q_{(12 \rightarrow c_4)}>0$ nodes 5 and 12 will join community $c_1$ and $c_4$, respectively. The subset of nodes for reevaluation will be $R=\{2,3,4,5,10,11,12,14,16,17,18,19,20\}$ which is displayed in red in part (d). }
        \label{figApp:VSexample}
\end{figure}

\begin{figure}[tbh]
    \centering
    \captionsetup{width=.9\linewidth}
    	\begin{subfigure}{0.32\textwidth}
		 \includegraphics[width=\textwidth]{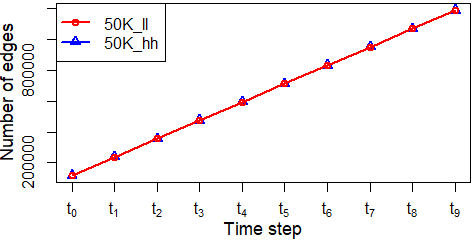}
		\caption{Number of edges for $50k\_ll$}
		\label{figApp:no_edges_50k}
	\end{subfigure}
	\begin{subfigure}{0.32\textwidth}
		 \includegraphics[width=\textwidth]{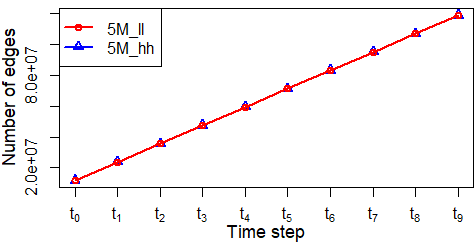}
		\caption{Number of edges for $5M\_ll$}
		\label{figApp:no_edges_5M}
	\end{subfigure}
	\begin{subfigure}{0.32\textwidth}
		 \includegraphics[width=\textwidth]{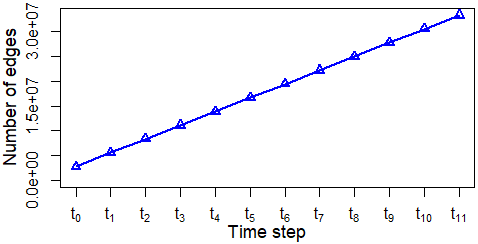}
		\caption {Number of edges per timestep for sx-stackoverflow}
		\label{figApp:no_edge_cite_sx}
	\end{subfigure}
	\begin{subfigure}{0.32\textwidth}
		 \includegraphics[width=\textwidth]{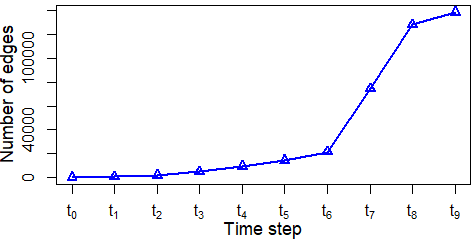}
		\caption {Number of edges per timestep for Arxiv HEP-TH}
		\label{figApp:no_edge_cite_sx}
	\end{subfigure}
	\caption{The change of input size per time step for synthetic and real world input}
\label{figApp:Inputsize}
\end{figure}

\begin{figure}[tbh]
    \centering
    \captionsetup{width=.9\linewidth}
	\begin{subfigure}{0.32\textwidth}
		 \includegraphics[width=\textwidth]{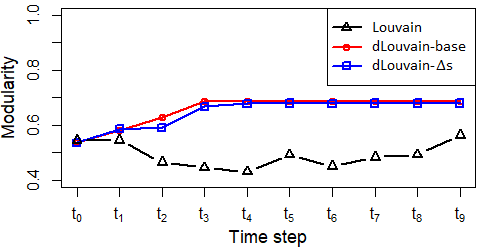}
		\caption{Modularity for $5M\_{hh}$ data}
		\label{figApp:mod_5M_hh}
	\end{subfigure}
	\begin{subfigure}{0.32\textwidth}
		 \includegraphics[width=\textwidth]{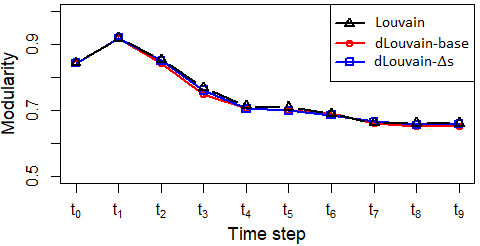}
		\caption{Modularity for Arxiv HEP-TH}
		\label{figApp:mod_cite_louv}
	\end{subfigure}
	\begin{subfigure}[b]{0.32\textwidth}
		 \includegraphics[width=\textwidth]{Modularity_citation.png}
		\caption{Modularity for sx-stackoverflow}
		\label{figApp:mod_sx_louv}
	\end{subfigure}
	\caption{Part (a) represents modularity for synthetic networks with a high level of overlap and high level of size variation between blocks (hh) of size 5M nodes. Part (b) represents modularity for Arxiv HEP-TH, and part (c) shows modularity for sx-stackoverflow.}
	\label{figApp:modularity}
\end{figure}

\end{document}